\documentclass{llncs}

\usepackage{amsmath}
\usepackage{amssymb}
\usepackage{stmaryrd}
\usepackage{dsfont}
\usepackage{graphicx}

\title{Cross-entropy optimisation of importance sampling parameters for statistical model checking}
\author{Cyrille J\'egourel, Axel Legay \and Sean Sedwards\thanks{To whom correspondence should be addressed.}}
\institute{INRIA Rennes - Bretagne Atlantique,\\
\email{\{cyrille.jegourel,axel.legay,sean.sedwards\}@inria.fr}}

\begin{document}
\maketitle
\begin{abstract}
Statistical model checking avoids the exponential growth of states associated with probabilistic model checking by estimating properties from multiple executions of a system and by giving results within confidence bounds. Rare properties are often very important but pose a particular challenge for simulation-based approaches, hence a key objective under these circumstances is to reduce the number and length of simulations necessary to produce a given level of confidence. Importance sampling is a well-established technique that achieves this, however to maintain the advantages of statistical model checking it is necessary to find good importance sampling distributions without considering the entire state space.

Motivated by the above, we present a simple algorithm that uses the notion of cross-entropy to find the optimal parameters for an importance sampling distribution. In contrast to previous work, our algorithm uses a low dimensional vector of parameters to define this distribution and thus avoids the often intractable explicit representation of a transition matrix. We show that our parametrisation leads to a unique optimum and can produce many orders of magnitude improvement in simulation efficiency. We demonstrate the efficacy of our methodology by applying it to models from reliability engineering and biochemistry.\end{abstract}

\section{Introduction}
The need to provide accurate predictions about the behaviour of complex systems is increasingly urgent. With computational power becoming ever-more affordable and compact, computational systems are inevitably becoming increasingly concurrent, distributed and adaptive, creating a correspondingly increased burden to check that they function correctly. At the same time, users expect high performance and reliability, prompting the need for equally high performance analysis tools and techniques.

The most common method to ensure the correctness of a system is by testing it with a number of test cases having predicted outcomes that can highlight specific problems. Testing techniques have been effective discovering bugs in many industrial applications and have been incorporated into sophisticated tools \cite{Godefroid2008}. Despite this, testing is limited by the need to hypothesise scenarios that may cause failure and the fact that a reasonable set of test cases is unlikely to cover all possible eventualities; errors and modes of failure may remain undetected and quantifying the likelihood of failure using a series of test cases is difficult.

Model checking is a formal technique that verifies whether a system satisfies a property specified in temporal logic under all possible scenarios. In recognition of non-deterministic systems and the fact that a Boolean answer is not always useful, {\em probabilistic} model checking quantifies the probability that a system satisfies a property. In particular, `numerical' (alternatively `exact') probabilistic model checking offers precise and accurate analysis by exhaustively exploring the state space of non-deterministic systems and has been successfully applied to a wide variety of protocols, algorithms and systems. The result of this technique is the exact (within limits of numerical precision) probability that a system will satisfy a property of interest, however the exponential growth of the state space limits its applicability. The typical $10^8$ state limit of exhaustive approaches usually represents an insignificant fraction of the state space of real systems that may have tens of orders of magnitude more states than the number of protons in the universe ($\sim10^{80}$).

Under certain circumstances it is possible to guarantee the performance of a system by specifying it in such a way that (particular) faults are impossible. Compositional reasoning and various symmetry reduction techniques can also be used to combat state-space explosion, but in general the size, unpredictability and heterogeneity of real systems \cite{Basu2010} make these techniques infeasible. Static analysis has also been highly successful in analysing and debugging software and other systems, although it cannot match the precision of quantitative analysis of dynamic properties achieved using probabilistic and stochastic temporal logic. 

While the state space explosion problem is unlikely to ever be adequately solved for all systems, simulation-based approaches are becoming increasingly tractable due to the availability of high performance hardware and algorithms. In particular, statistical model checking (SMC) combines the simplicity of testing with the formality and precision of numerical model checking; the core idea being to create multiple independent execution traces of the system and individually verify whether they satisfy some given property. By modelling the executions as a Bernoulli random variable and using advanced statistical techniques, such as Bayesian inference \cite{Jha2009} and hypothesis testing \cite{Younes2002}, the results are combined in an efficient manner to decide whether the system satisfies the property with some level of confidence, or to estimate the probability that it does.
Knowing a result with less than 100\% confidence is often sufficient in real applications, since the confidence bounds may be made arbitrarily tight. Moreover, statistical model checking may offer the only feasible means of quantifying the performance of many complex systems. Evidence of this is that statistical model checking has been used to find bugs in large, heterogeneous aircraft systems \cite{Basu2010}. Notable statistical model checking platforms include APMC \cite{Herault2004}, YMER \cite{Younes2005} and VESTA \cite{Sen2005}. Well-established numerical model checkers, such as PRISM \cite{Kwiatkowska2002} and UPPAAL \cite{Bengtsson1996}, are now also including statistical model checking engines to cope with larger models.

A key challenge facing statistical model checking is to reduce the length (steps and cpu time) and number of simulation traces necessary to achieve a result with given confidence. The current proliferation of parallel computer architectures (multiple cpu cores, grids, clusters, clouds and general purpose computing on graphics processors, etc.) favours statistical model checking by making the production of multiple independent simulation runs relatively easy. Despite this, certain models still require a large number of simulation steps to verify a property and it is thus necessary to make simulation as efficient as possible. Rare (unlikely) properties pose a particular problem for simulation-based approaches, since they are not only difficult to observe (by definition) but their probability is difficult to bound \cite{Heidelberger1995}.

The term `rare event' is ubiquitous in the literature, but here we specifically consider rare {\em properties} defined in temporal logic. This distinguishes rare states from rare paths that may or may not contain rare states. The distinction does not significantly alter the mathematical derivation of our algorithm, however it affects the applicability of simple heuristics that are able to find (reasonably) good importance sampling distributions. This is of relevance because our algorithm works by a process of iterative refinement, starting from an initial distribution that must produce at least a few traces that satisfy the property. In what follows we consider discrete space Markov models and present a simple algorithm to find an optimal set of importance sampling parameters, using the concept of minimum cross-entropy \cite{Kullback1968,Shore1980}. Our parametrisation arises naturally from the syntactic description of the model and thus constitutes a low dimensional vector in comparison to the state space of the model. We show that this parametrisation has a unique optimum and demonstrate its effectiveness on reliability and (bio)chemical models. We describe the advantages and potential pitfalls of our approach and highlight areas for future research.

\section{Importance sampling}

Our goal is to estimate the probability of a property by simulation and bound the error of our estimation. When the property is not rare there are standard bounding formulae (e.g., the Chernoff and Hoeffding bounds \cite{Chernoff1952,Hoeffding1963}) that relate absolute error, confidence and the required number of simulations to achieve them, {\em independent} of the probability of the property. As the property becomes rarer, however, absolute error ceases to be useful and it is necessary to consider relative error, defined as the standard deviation of the estimate divided by its expectation. With Monte Carlo simulation relative error is unbounded with increasing rarity \cite{Rubino2009}, but it is possible to bound the error by means of importance sampling \cite{Shahabuddin1994,Heidelberger1995}.

Importance sampling is a technique that can improve the efficiency of simulating rare events and has been receiving considerable interest of late in the field of statistical model checking (e.g., \cite{ClarkeZuliani2011,Barbot2012}). It works by simulating under an (importance sampling) distribution that makes a property more likely to be seen and then uses the results to calculate the probability under the original distribution by compensating for the differences. The concept arose from work on the `Monte Carlo method' \cite{MetropolisUlam1949} in the Manhattan project during the 1940s and was originally used to quantify the performance of materials and solve otherwise intractable analytical problems with limited computer power (see, e.g., \cite{Kahn1949}). For importance sampling to be effective it is necessary to define a `good' importance sampling distribution: (i) the property of interest must be seen frequently in simulations and (ii) the distribution of the paths that satisfy the property in the importance sampling distribution must be as close as possible to the distribution of the same paths in the original distribution (up to a normalising factor). The literature in this field sometimes uses the term `zero variance' to describe an optimal importance sampling distribution, referring to the fact that with an optimum importance sampling distribution all simulated paths satisfy the property and the estimator has zero variance. It is important to note, however, that a sub-optimal distribution may meet requirement (i) without necessarily meeting requirement (ii). Failure to consider (ii) can result in gross errors and overestimates of confidence (e.g. a distribution that simulates just one path that satisfies the given property). The algorithm we present in Section \ref{sec:algorithm} addresses both (i) and (ii).

Importance sampling schemes fall into two broad categories: {\em state dependent tilting} and {\em state independent tilting} \cite{DeBoer2000}. State dependent tilting refers to importance sampling distributions that individually bias (`tilt') every transition probability in the system. State independent tilting refers to importance sampling distributions that change classes of transition probabilities, independent of state. The former offers greatest precision but is infeasible for large models. The latter is more tractable but may not produce good importance sampling distributions. Our approach is a kind of {\em parametrised tilting} that potentially affects all transitions differently, but does so according to a set of parameters.

\subsection{Estimators}
Let $\Omega$ be a probability space of paths, with $f$ a probability density function over $\Omega$ and $z(\omega)\in\{0,1\}$ a function indicating whether a path $\omega$ satisfies some property $\phi$. In the present context, $z$ is defined by a formula of an arbitrary temporal logic over execution traces. The expected probability $\gamma$ that $\phi$ occurs in a path is then given by
\begin{equation}\gamma = \int_\Omega z(\omega)f(\omega)\;\textnormal d\omega\label{gamma}\end{equation}
and the standard Monte Carlo estimator of $\gamma$ is given by
$$\tilde{\gamma}=\frac{1}{N_\textnormal{\tiny MC}}\sum^{N_\textnormal{\tiny MC}}_{i=1}z(\omega_i)$$
$N_\textnormal{\tiny MC}$ denotes the number of simulations used by the Monte Carlo estimator and note that $z(\omega_i)$ is effectively the realisation of a Bernoulli random variable with parameter $\gamma$. Hence $\textnormal{Var}(\tilde{\gamma})=\gamma(1-\gamma)$ and for $\gamma\rightarrow 0$, $\textnormal{Var}(\tilde{\gamma})\approx\gamma$. 
Let $f'$ be another probability density function over $\Omega$, absolutely continuous with $zf$, then Equation \eqref{gamma} can be written
$$\gamma = \int_\Omega z(\omega)\frac{f(\omega)}{f'(\omega)}f'(\omega)\;\textnormal d\omega$$
$L=f/f'$ is the {\em likelihood ratio} function, so
\begin{equation}\gamma = \int_\Omega L(\omega)z(\omega)f'(\omega)\;\textnormal d\omega\label{gammaL}\end{equation}
We can thus estimate $\gamma$ by simulating under $f'$ and compensating by $L$:
$$\tilde{\gamma}=\frac{1}{N_\textnormal{\tiny IS}}\sum^{N_\textnormal{\tiny IS}}_{i=1} L(\omega_i)z(\omega_i)$$
$N_\textnormal{\tiny IS}$ denotes the number of simulations used by the importance sampling estimator.
The goal of importance sampling is to reduce the variance of the rare event and so achieve a narrower confidence interval than the Monte Carlo estimator, resulting in $N_\textnormal{\tiny IS}\ll N_\textnormal{\tiny MC}$. In general, the importance sampling distribution $f'$ is chosen to produce the rare property more frequently, but this is not the only criterion. The optimal importance sampling distribution, denoted $f^*$ and defined as $f$ conditioned on the rare event, produces only traces satisfying the rare property:
\begin{equation}f^* = \frac{zf}{\gamma}\label{optimal}\end{equation}
This leads to the term `zero variance estimator' with respect to $Lz$, noting that, in general, $\textnormal{Var}(f^*)\geq 0$.

In the context of statistical model checking $f$ usually arises from the specifications of a model described in some relatively high level language. Such models do not, in general, explicitly specify the probabilities of individual transitions, but do so implicitly by parametrised functions over the states. We therefore consider a class of models that can be described by guarded commands \cite{Dijkstra1975} extended with stochastic rates. Our parametrisation is a vector of strictly positive values $\lambda\in(\mathbb{R}^+)^n$ that multiply the stochastic rates and thus maintain the absolutely continuous property between distributions. Note that this class includes both discrete and continuous time Markov chains and that in the latter case our mathematical treatment works with the embedded discrete time process.

In what follows we are therefore interested in parametrised distributions and write $f(\cdot,\lambda)$, where $\lambda=\{\lambda_1,\dots,\lambda_n\}$ is a vector of parameters, and distinguish different density functions by their parameters. In particular, $\mu$ is the original vector of the model and $f(\cdot,\mu)$ is therefore the original density. We can thus rewrite Equation \eqref{gammaL} as
$$\gamma = \int_\Omega L(\omega)z(\omega)f(\omega,\lambda)\;\textnormal d\omega$$
where $L(\omega)=f(\omega,\mu)/f(\omega,\lambda)$. We can also rewrite Equation \eqref{optimal}
$$f^* = \frac{zf(\cdot,\mu)}{\gamma}$$
and write for the optimal parametrised density $f(\cdot,\lambda^*)$.
We define the optimum parametrised density function as the density that minimises the {\em cross-entropy} \cite{Kullback1968} between $f(\cdot,\lambda)$ and $f^*$ for a given parametrisation and note that, in general, $f^*\neq f(\cdot,\lambda^*)$.
\subsection{The cross-entropy method}
Cross-entropy \cite{Kullback1968} (alternatively {\em relative entropy} or Kullback-Leibler divergence) has been shown to be an effective directed measure of distance between distributions \cite{Shore1980}. With regard to the present context, it has also been shown to be useful in finding optimum distributions for importance sampling \cite{Rubinstein1999,DeBoer2000,Ridder2005}.

Given two probability density functions $f$ and $f'$ over the same probability space $\Omega$, the cross-entropy from $f$ to $f'$ is given by
\begin{eqnarray}
\textnormal{CE}(f,f') = \int_\Omega f(\omega)\log\frac{f(\omega)}{f'(\omega)}\,\textnormal{d}\omega &=& \int_\Omega f(\omega)\log f(\omega) - f(\omega)\log f'(\omega)\,\textnormal{d}\omega\nonumber\\
& = &\textnormal{H}(f) - \int_\Omega f(\omega)\log f'(\omega)\,\textnormal{d}\omega\label{KLdivergence}
\end{eqnarray}
where $\textnormal{H}(f)$ is the entropy of $f$. To find $\lambda^*$ we minimise $\textnormal {CE}(z(\omega)f(\omega,\mu),f(\omega,\lambda))$, noting that $\textnormal{H}(f(\omega,\mu))$ is independent of $\lambda$:
\begin{equation}\lambda^*=\mathop{\arg\max}_\lambda\int_\Omega z(\omega)f(\omega,\mu)\log f(\omega,\lambda)\,\textnormal d\omega\label{argmax}\end{equation}
Estimating $\lambda^*$ directly using Equation \eqref{argmax} is hard, so we re-write it using importance sampling density $f(\cdot,\lambda')$ and likelihood ratio function $L(\omega) = f(\omega,\mu)/f(\omega,\lambda')$:
\begin{equation}\lambda^*=\mathop{\arg\max}_\lambda\int_\Omega z(\omega) L(\omega)f(\omega,\lambda')\log f(\omega,\lambda)\,\textnormal d\omega\label{argmaxL}\end{equation}
Using Equation \eqref{argmaxL} we can construct an unbiased importance sampling estimator of $\lambda^*$ and use it as the basis of an iterative process to obtain successively better estimates:
\begin{equation}\tilde{\lambda^*}=\lambda^{(j+1)}=\mathop{\arg\max}_\lambda\sum_{i=1}^{N_j} z(\omega_i^{(j)})L^{(j)}(\omega_i^{(j)})\log f(\omega_i^{(j)},\lambda)\label{argmaxE}\end{equation}
$N^j$ is the number of simulation runs on the $j$\textsuperscript{th} iteration, $\lambda^{(j)}$ is the $j$\textsuperscript{th} set of estimated parameters, $L^{(j)}(\omega)= f(\omega,\mu)/f(\omega,\lambda^{(j)})$ is the $j$\textsuperscript{th} likelihood ratio function, $\omega_i^{(j)}$ is the $i$\textsuperscript{th} path generated using $f(\cdot,\lambda^{(j)})$ and $f(\omega_i^{(j)},\lambda)$ is the probability of path $\omega_i^{(j)}$ under the distribution $f(\cdot,\lambda^{(j)})$.

\section{A parametrised cross-entropy algorithm}\label{sec:algorithm}

We consider a system of $n$ guarded commands with vector of rate functions $K=(K_1,\dots,K_n)$ and corresponding vector of parameters $\lambda=(\lambda_1,\dots,\lambda_n)$. In any given state the probability that command $k$ is chosen is given by
$$\frac{\lambda_k K_k}{\langle K,\lambda\rangle}$$
where the notation $\langle\cdot,\cdot\rangle$ denotes a scalar product. For the purposes of simulation we consider a space of finite paths $\omega\in\Omega$. Let $U_k(\omega)$ be the number of transitions of type $k$ occurring in $\omega$. We therefore have
$$f(\omega,\lambda)=\prod_{k}^{n} (\lambda_k)^{U_k(\omega)} \prod_{s=1}^{U_k(\omega)}\frac{K^s_k}{\langle K^s,\lambda\rangle}$$
Here vector $K$ is indexed by $s$ to emphasise its state-dependence. The likelihood ratios are thus of the form
$$L^{(j)}(\omega)=\prod_{k}^{n}\left( \frac{\mu_k}{\lambda^{(j)}_k}\right)^{U_k(\omega)} \prod_{s=1}^{U_k(\omega)}\frac{\langle K^s,\lambda^{(j)}\rangle}{\langle K^s,\mu\rangle}$$
We substitute these expressions in the cross-entropy estimator Equation (\ref{argmaxE}) and for compactness substitute $z_i=z(\omega_i)$, $u_i(k)=U_k(\omega_i)$ and $l_i=L^{(j)}(\omega_i)$ to get
\begin{eqnarray}
& & \mathop{\arg\max}_{\lambda}\sum_{i=1}^{N} l_iz_i\log\prod_{k}^{n} \left(\lambda_k^{u_i(k)} \prod_{s=1}^{u_i(k)}\frac{K^{i,s}_k}{\langle K^{i,s},\lambda\rangle}\right)\\
\nonumber{} &=&\mathop{\arg\max}_{\lambda} \sum_{i=1}^{N} \sum_k^n l_iz_iu_i(k)\left(\mbox{log}(\lambda_k) + \sum_{s=1}^{u_i(k)}\mbox{log}(K^{i,s}_k) -\sum_{s=1}^{u_i(k)}\mbox{log}(\langle K^{i,s},\lambda\rangle)\right)
\end{eqnarray}
We partially differentiate with respect to $\lambda_k$ and get the non-linear system
\begin{equation}
\frac{\partial{f}}{\partial{\lambda_k}}(\lambda)=0 \Leftrightarrow \sum_{i=1}^{N} l_iz_i\left(\frac{u_i(k)}{\lambda_k} - \sum_{s=1}^{|\omega_i|}\frac{K_k^{i,s}}{\langle K^{i,s},\lambda\rangle}\right)=0\label{eq:partial}
\end{equation}
where $|\omega_i|$ is the length of the path $\omega_i$.
\begin{theorem}A solution of Equation (\ref{eq:partial}) is almost surely a maximum, up to a normalising scalar.
\end{theorem}
\begin{proof}
Using a standard result, it is sufficient to show that the Hessian matrix in $\lambda$ is negative semi-definite. Consider $f_i$:
$$f_i(\lambda)=\sum_k u_i(k)\left(\mbox{log}(\lambda_k) + \sum_{s=1}^{u_i(k)}\mbox{log}(K^{i,s}_k) -\sum_{s=1}^{u_i(k)}\mbox{log}(\langle K^{i,s},\lambda\rangle)\right)$$
The Hessian matrix in $\lambda$ is of the following form with $v_k^{(s)}=\frac{K_k^{(s)}}{\langle K^{(s)},\lambda\rangle}$ and $v_k=(v_k^{(s)})_{1\leq s\leq U_k(\omega)}$:
$$H_i=G-D$$
where $G=(g_{kk'})_{1\leq k,k'\leq n}$ is the following Gram matrix
$$g_{kk'}=\langle v_k, v_{k'}\rangle$$
and $D$ is a diagonal matrix such that
$$d_{kk}=\frac{u_k}{\lambda_k^2}.$$
Note that asymptotically $d_{kk}=\frac{1}{\lambda_k}\sum_{s=1}^N v_k^{(s)}.$
We write $\mathds{1}_N=(1,\dots,1)$ for the vector of $N$ elements $1$, hence
$$d_{kk}=\frac{1}{\lambda_k}\langle v_k, \mathds{1}_N \rangle.$$
Furthermore, $\forall s, \sum_{k=1}^n\lambda_k v_k^{(s)} = 1$.
So, $\sum_{k'=1}^n\lambda_{k'} v_{k'} = \mathds{1}_N$.
Finally,
$$d_{kk}=\sum_{k'=1}^n \frac{\lambda_{k'}}{\lambda_k}\langle v_k, v_{k'}\rangle.$$
Let $x\in\mathbb{R}^n$. To prove the theorem we need to show that $-x^tHx\geq 0$.
\begin{eqnarray}
-x^tHx &= x^tDx - x^tGx\\
\nonumber{} &=&\sum_{k,k'} \frac{\lambda_{k'}}{\lambda_k}\langle v_k, v_{k'}\rangle x_k^2 - \sum_{k,k'} \langle v_k, v_{k'}\rangle x_kx_{k'}\\
\nonumber{} &=&\sum_{k<k'}\left( \left[ \frac{\lambda_{k'}}{\lambda_k}x_k^2 + \frac{\lambda_k}{\lambda_{k'}}x_{k'}^2 - 2x_kx_{k'}\right] \langle v_k, v_{k'}\rangle\right)\\
\nonumber{} &=&\sum_{k<k'}\left( \sqrt{\frac{\lambda_{k'}}{\lambda_k}}x_k - \sqrt{\frac{\lambda_k}{\lambda_{k'}}}x_{k'}\right)^2  \langle v_k, v_{k'}\rangle\\
\nonumber{} &\geq& 0
\end{eqnarray}
The Hessian matrix $H$ of $f$ is of the general form
$$H=\sum_{i=1}^{N}l_iz_iH_i$$
which is a positively weighted sum of non-positive matrices.
\qed\end{proof}
The Hessian is negative {\em semi}-definite because if $\lambda$ is a solution then $x\lambda, x\in\mathbb{R}^+$, is also a solution. The fact that there is a unique optimum, however, makes it conceivable to find $\lambda^*$ using standard optimising techniques such as Newton and quasi-Newton methods. To do so would require introducing a suitable normalising constraint in order to force the Hessian to be negative definite. In the case of the cross-entropy algorithm of \cite{Ridder2005}, this constraint is inherent because it works at the level of individual transition probabilities that sum to 1 in each state. We note here that in the case that our parameters apply to individual transitions, such that one parameter corresponds to exactly one transition, Equation (\ref{eq:algorithm}) may be transformed to Equation (9) of \cite{Ridder2005} by constraining $\langle K,\lambda\rangle=1$. Equation (9) of \cite{Ridder2005} has been shown in \cite{Ridder2010} to converge to $f^*$, implying that under these circumstances $f(\cdot,\lambda^*)=f^*$ and that it may be possible to improve our parametrised importance sampling distribution by increasing the number of parameters.
\subsection{The algorithm}
Equation (\ref{eq:partial}) leads to the following expression for $\lambda_k$:
\begin{equation}
\lambda_k=\frac{\sum_{i=1}^{N} l_iz_iu_i(k)}{\sum_{i=1}^{N} l_iz_i\sum_{s=1}^{|\omega_i|}\frac{K_k^s}{\langle K^s,\lambda\rangle}}\label{eq:lambdak}
\end{equation}
In this form the expression is not useful because the right hand side is dependent on $\lambda_k$ in the scalar product. Hence, in contrast to update formulae based on unbiased estimators, as given by Equation (\ref{argmaxE}) and in \cite{Ridder2005,DeBoer2000}, we construct an iterative process based on a biased estimator but having a fixed point that is the optimum:
\begin{equation}
\lambda_k^{(j+1)}=\frac{\sum_{i=1}^{N_j} l_iz_iu_i(k)}{\sum_{i=1}^{N_j} l_iz_i\sum_{s=1}^{|\omega_i|}\frac{K_k^s}{\langle K^s,\lambda^{(j)}\rangle}}\label{eq:algorithm}
\end{equation}
Equation (\ref{eq:algorithm}) can be seen as an implementation of Equation (\ref{eq:lambdak}) that uses the previous estimate of $\lambda$ in the scalar product, however it works by reducing the distance between successive distributions, rather than by explicitly reducing the distance from the optimum. We offer no proof of convergence here, but assert that if it converges it converges to $\lambda^*$.

To use the algorithm it is necessary to start with an initial simulation distribution $f(\cdot,\lambda^{(0)})$ and number of simulations $N_0$ that produce at least a few traces that satisfy the property. The choice of $\lambda^{(0)}$ and $N_0$ is highly dependent on the model and the property and must in general be established by heuristics or trial and error. When the number of parameters is small and the property is very rare, an effective strategy is to iterate the algorithm with relatively low $N_0$ and random parameters until a suitable trace is observed. If the model and property are similar to a previous combination for which parameters were found, those parameters are likely to provide a good initial estimate. Increasing the parameters associated to obviously small rates may help (along the lines of simple failure biasing \cite{Shahabuddin1994}), however the rareness of a property expressed in temporal logic may not always be related to low probabilities. Finding good initial distributions for arbitrary systems and temporal properties is the subject of ongoing work.

Given a sufficient number of traces from the first iteration, Equation \ref{eq:algorithm} should provide a better set of parameters. The expected behaviour is that on successive iterations the number of traces that satisfy the property increases, however it is important to note that the algorithm optimises the {\em quality} of the distribution and that the number of traces that satisfy the property is merely emergent of that. As has been noted, in general $f(\cdot,\lambda^*)\neq f^*$, hence it is likely that fewer than $100\%$ of traces will satisfy the property when simulating under $f(\cdot,\lambda^*)$. One consequence of this is that an initial set of parameters may produce more traces that satisfy the property than the final set (see, e.g., Figure \ref{fig:trues470}).

It is conceivable that certain guarded commands play no part in traces that satisfy the property, in which case Equation (\ref{eq:algorithm}) would make the corresponding parameter zero with no adverse effects. It is also conceivable that an important command is not seen on a particular iteration, but making its parameter zero would prevent it being seen on any subsequent iteration. To avoid this it is necessary to adopt a `smoothing' strategy \cite{Ridder2005} that reduces the significance of an unseen command without setting it to zero. The strategy adopted for the examples shown below is to divide the parameter of unseen commands by two. The effects of this can be seen in Figure \ref{fig:lambdaRidder3}. An alternative approach is to add a small fraction of the initial (or previous) parameters to every new parameter estimate. Whatever the strategy, since the parameters are unconstrained it is advisable to normalise them after each iteration (i.e., $\sum_k\lambda_k=\textnormal{const.}$) in order to judge progress.

Once the parameters have converged it is then possible to perform a final set of simulations to estimate the probability of the rare property. The usual assumption is that $N_j\ll N_\textnormal{\tiny IS}\ll N_\textnormal{\tiny MC}$, however it is often the case that parameters converge fast, so it is expedient to use some of the simulation runs generated during the course of the optimisation as part of the final estimation.

\section{Examples}
The following examples are included to illustrate the performance of our algorithm and parametrisation. The first is an example of a chemical system, often used to motivate stochastic simulation, while the second is a standard repair model. All simulations were performed using our statistical model checking platform, PLASMA \cite{Jegourel2012}.
\subsection{Chemical network}
Following the success of the human genome project, with vast repositories of biological pathway data available online, there is an increasing expectation that formal methods can be applied to biological systems. The network of chemical reactions given below is abstract but typical of biochemical systems and demonstrates the potential of statistical model checking to handle the enormous state spaces of biological models. In particular, we demonstrate the efficacy of our algorithm by applying it to quantify two rare dynamical properties of the system.

We consider a well stirred chemically reacting system comprising five reactants (molecules of type $A,B,C,D,E$), a dimerisation reaction (\ref{dimerisation}) and two decay reactions (\ref{decayC},\ref{decayD}):
\begin{eqnarray}
A+B\stackrel{k_1}{\rightarrow} C\label{dimerisation}\\
C\stackrel{k_2}{\rightarrow} D\label{decayC}\\
D\stackrel{k_3}{\rightarrow} E\label{decayD}
\end{eqnarray}
Under the assumption that the molecules move randomly and that elastic collisions significantly outnumber unreactive, inelastic collisions, the system may be simulated using mass action kinetics as a continuous time Markov chain \cite{Gillespie1977}. The semantics of Equation (\ref{dimerisation}) is that if a molecule of type $A$ encounters a molecule of type $B$ they will combine to form a molecule of type $C$ after a delay drawn from an exponential distribution with mean $k_1$. The decay reactions have the semantics that a molecule of type $C$ ($D$) spontaneously decays to a molecule of type $D$ ($E$) after a delay drawn from an exponential distribution with mean $k_2$ ($k_3$). The reactions (\ref{dimerisation},\ref{decayC},\ref{decayD}) are modelled by three guarded commands having importance sampling parameters $\lambda_1,\lambda_2$ and $\lambda_3$, respectively. A typical simulation run is illustrated in Figure \ref{fig:chemsim}, where the x-axis is steps rather than time to aid clarity. $A$ and $B$ combine rapidly to form $C$ which peaks before decaying slowly to $D$. The production of $D$ also peaks while $E$ rises monotonically.
\begin{figure}[ht]
\begin{minipage}[t]{0.455\linewidth}
\centering
\includegraphics[scale=0.5]{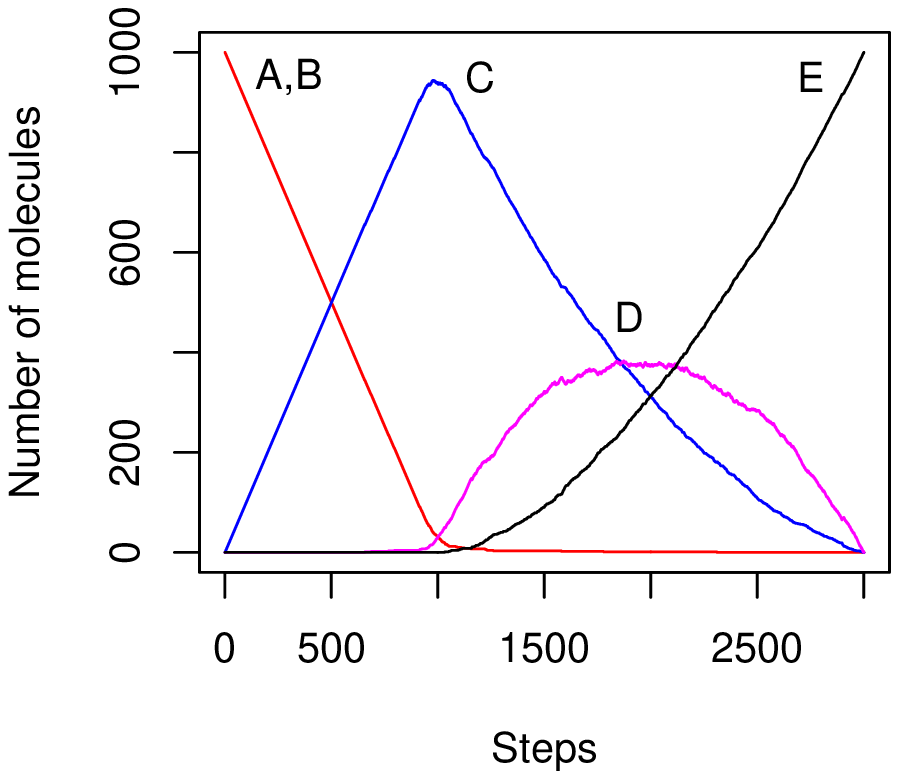}
\caption{A typical stochastic simulation trace of reactions (\ref{dimerisation}-\ref{decayD}).}
\label{fig:chemsim}
\end{minipage}
\hspace{0.5cm}
\begin{minipage}[t]{0.455\linewidth}
\centering
\includegraphics[scale=0.5]{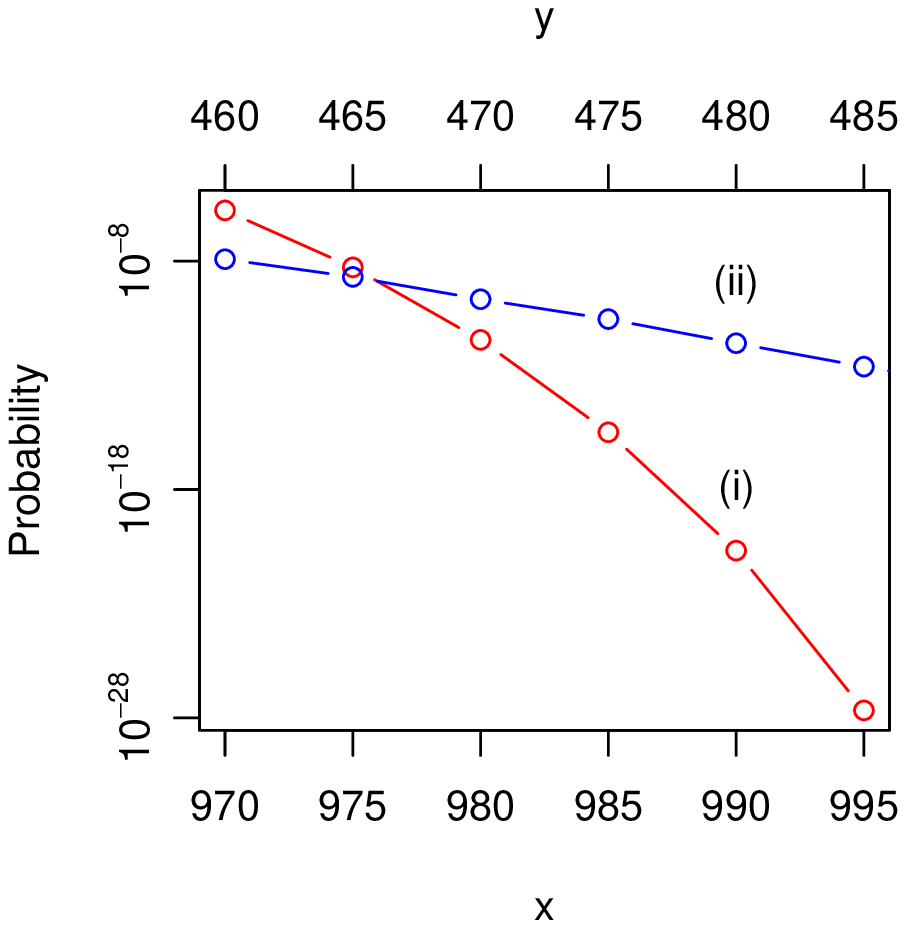}
\caption{(i) Pr$[\Diamond\,C\geq x]$ (ii) Pr$[\Diamond\,D\geq y]$}
\label{fig:chemresults}
\end{minipage}
\end{figure}

With an initial vector of molecules $(1000,1000,0,0,0)$, corresponding to types $(A,B,C,D,E)$, the state space contains $\sim 10^{15}$ states. We know from a static analysis of the reactions that it is possible for the numbers of molecules of $C$ and $D$ to reach the initial number of $A$ and $B$ molecules (i.e., 1000) and that this is unlikely. To find out exactly how unlikely we consider the probabilities of the following rare properties defined in linear temporal logic: (i) $\Diamond\,C\geq x, x\in\{970,975,980,985,990,995\}$ and (ii) $\Diamond\,D\geq y, y\in\{460,465,470,475,480,485\}$. The results are plotted in Figure \ref{fig:chemresults}.

To establish an initial distribution our algorithm (Equation (\ref{eq:algorithm})) was iterated with random parameters and $N_0=1000$ until one or more traces were observed that satisfied the property in question. No more than 10 such iterations were needed in any case. The algorithm was then iterated 20 times using $N_j=1000$. Despite the large state space, this value of $N_j$ was found to be sufficient to produce reliable results. Starting from randomly chosen initial values, the convergence of parameters can be seen in Figure \ref{fig:lambda470}. Figure \ref{fig:trues470} illustrates that the number of paths satisfying a property can actually decrease as the quality of the distribution improves. Figure \ref{fig:gamma470} illustrates the convergence of the estimate and sample variance using the importance sampling parameters generated during the course of running the algorithm. The initial set of parameters appear to give a very low variance, however this is clearly erroneous with respect to subsequent values. Noting that the variance of standard Monte Carlo simulation of rare events gives a variance approximately equal to the probability and assuming that the sample variance is close to the true variance, Figure \ref{fig:gamma470} suggests that we have made a variance reduction of approximately $10^7$. 
\begin{figure}[ht]
\begin{minipage}[t]{0.455\linewidth}
\centering
\includegraphics[scale=0.5]{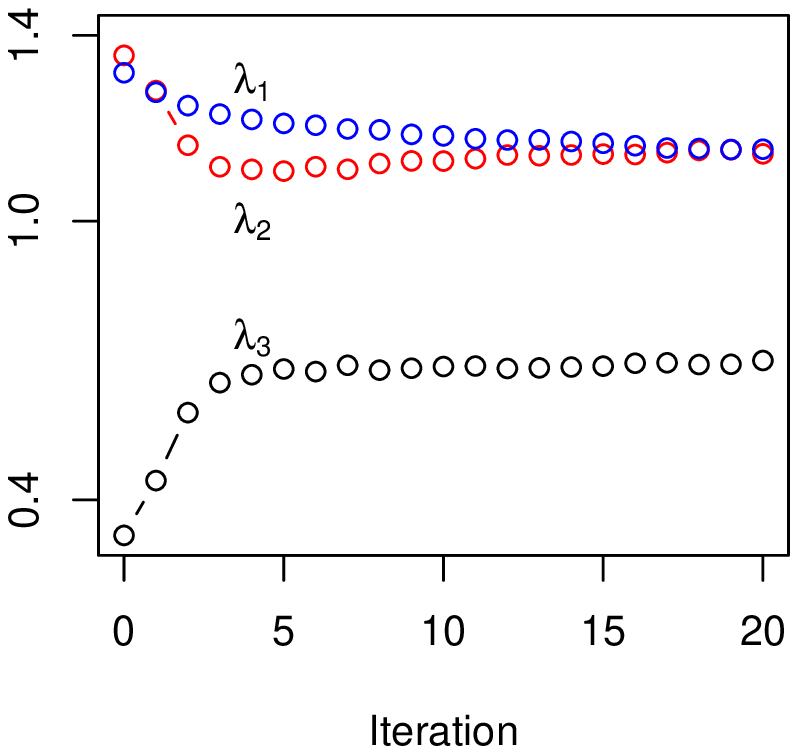}
\caption{Convergence of parameters for $\Diamond\,D\geq 470$ in the chemical model using $N_j=1000$.}
\label{fig:lambda470}
\end{minipage}
\hspace{0.5cm}
\begin{minipage}[t]{0.455\linewidth}
\centering
\includegraphics[scale=0.5]{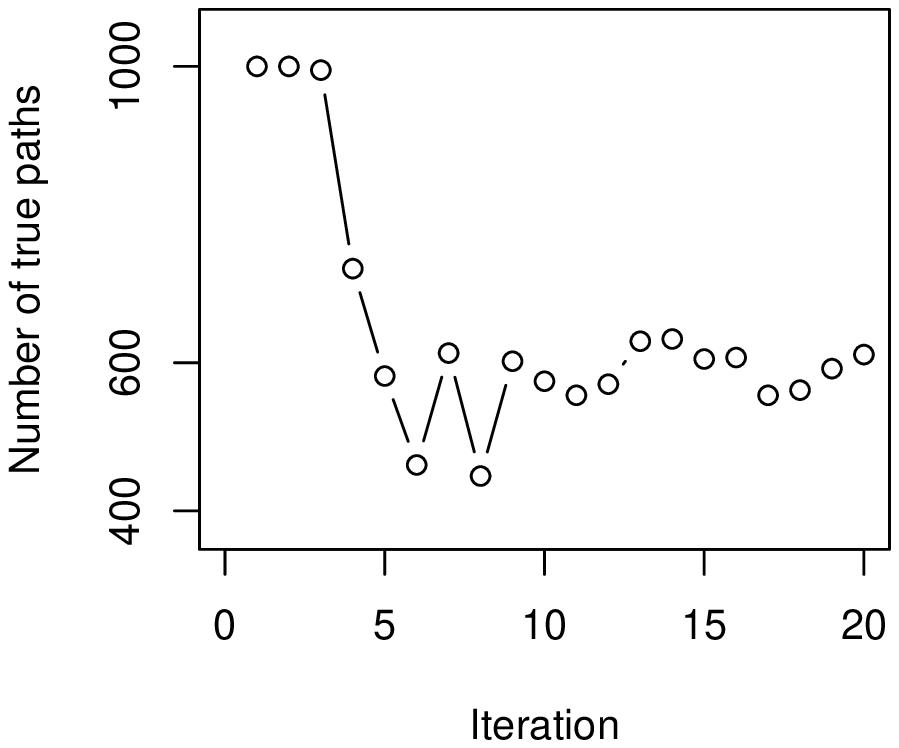}
\caption{Convergence of number of paths satisfying $\Diamond\,D\geq 470$ in the chemical model using $N_j=1000$.}
\label{fig:trues470}
\end{minipage}
\end{figure}

\begin{figure}[ht]
\begin{minipage}[t]{0.455\linewidth}
\centering
\includegraphics[scale=0.5]{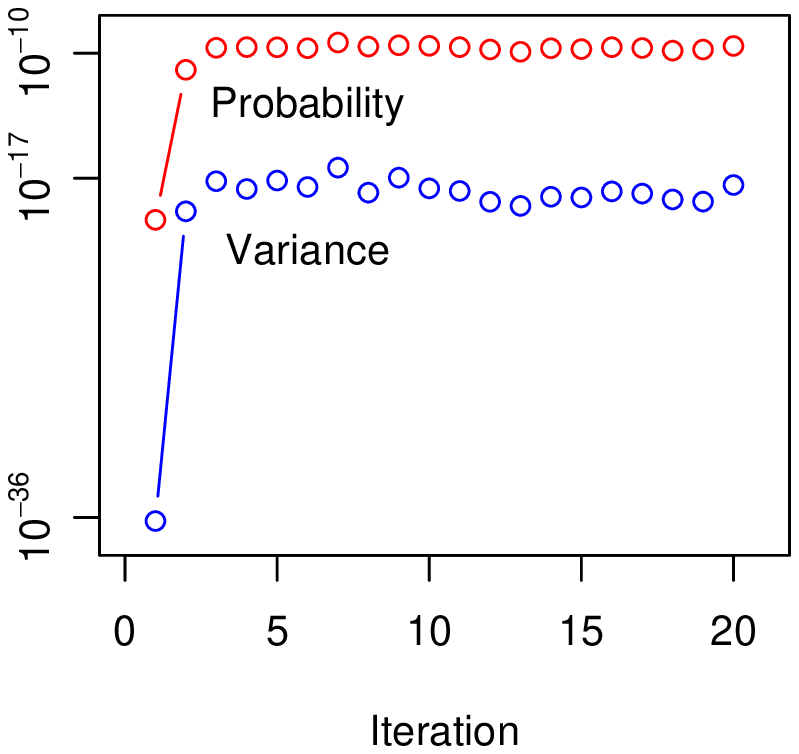}
\caption{Convergence of probability and sample variance for $\Diamond\,D\geq 470$ in the chemical model using $N_j=1000$.}
\label{fig:gamma470}
\end{minipage}
\hspace{0.5cm}
\begin{minipage}[t]{0.455\linewidth}
\centering
\includegraphics[scale=0.5]{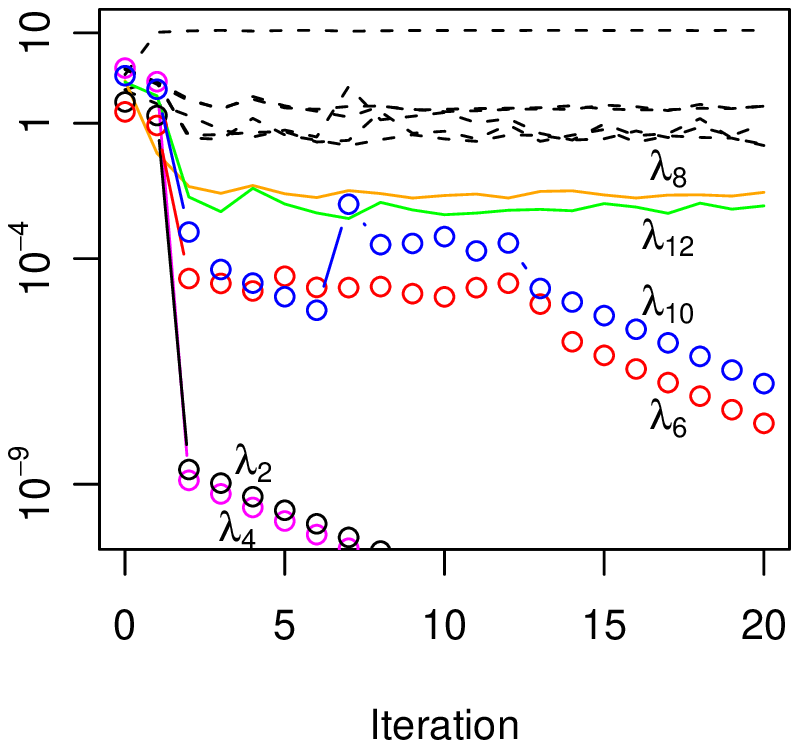}
\caption{Convergence of parameters (dashed/solid lines) and effect of smoothing strategy (circles) in repair model using $N_j=10000$.}
\label{fig:lambdaRidder3}
\end{minipage}
\end{figure}

\subsection{Repair model}
To a large extent the need to certify system reliability motivates the use of formal methods and thus reliability models are studied extensively in the literature. The following example is taken from \cite{Ridder2005} and features a moderately large state space of 40,320 states that can be investigated using numerical methods to corroborate our results.

The system is modelled as a continuous time Markov chain and comprises six types of subsystems $(1,\dots,6)$ containing, respectively, $(5,4,6,3,7,5)$ components that may fail independently. The system's evolution begins with no failures and with various probabilistic rates the components fail and are repaired. The failure rates are $(2.5\epsilon,\epsilon,5\epsilon,3\epsilon,\epsilon,5\epsilon)$, $\epsilon=0.001$, and the repair rates are $(1.0,1.5,1.0,2.0,1.0,1.5)$, respectively. Each subsystem type is modelled by two guarded commands: one for failure and one for repair. The property under investigation is the probability of a complete failure of a subsystem (i.e., the failure of all components of one type), given an initial condition of no failures. This can be expressed in temporal logic as Pr[X($\neg${\em init} U {\em failure})].

Figure \ref{fig:lambdaRidder3} shows the convergence of parameters (dashed/solid lines) and highlights the effects of the adopted smoothing strategy (circles). Parameters $\lambda_2$ and $\lambda_4$ (the parameters for repair commands of types 1 and 2, respectively) are attenuated from the outset by the convergence of the other parameters (because of the normalisation). Once their values are small relative to the normalisation constant (12 in this case), their corresponding commands no longer occur and their values experience exponential decay as a result of smoothing (division by two at every subsequent step). Parameters $\lambda_6$ and $\lambda_{10}$ (the parameters for repair commands of types 3 and 5, respectively) converge for 12 steps but then also decay. The parameters for the repair commands of types 4 and 6 (solid lines) are the smallest of the parameters that converge. The fact that the repair transitions are made less likely by the algorithm agrees with the intuition that we are interested in direct paths to failure. The fact that they are not necessarily made zero reinforces the point that the algorithm seeks to consider {\em all} paths to failure, including those that have intermediate repairs.

Figure \ref{fig:truesRidder3} plots the number of paths satisfying X($\neg${\em init} U {\em failure}) and suggests that for this model the parametrised distribution is close to the optimum. Figure \ref{fig:gammaRidder3} plots the estimated probability and sample variance during the course of the algorithm and superimposes the true probability calculated by PRISM \cite{PRISMwebsite}. The long term average agrees well with the true value (an error of -1.7\%, based on an average excluding the first two estimates), justifying our use of the sample variance as an indication of the efficacy of the algorithm: our importance sampling parameters provide a variance reduction of more than $10^5$.

\begin{figure}[ht]
\begin{minipage}[t]{0.455\linewidth}
\centering
\includegraphics[scale=0.5]{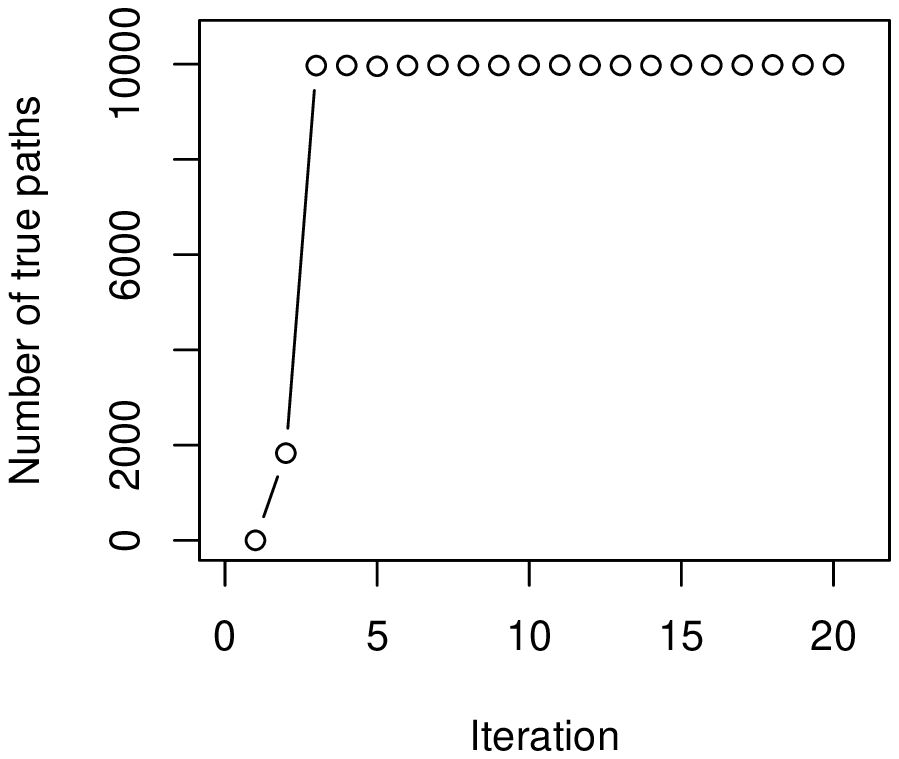}
\caption{Convergence of number of paths satisfying X($\neg${\em init} U {\em failure}) in the repair model using $N_j=10000$.}
\label{fig:truesRidder3}
\end{minipage}
\hspace{0.5cm}
\begin{minipage}[t]{0.455\linewidth}
\centering
\includegraphics[scale=0.5]{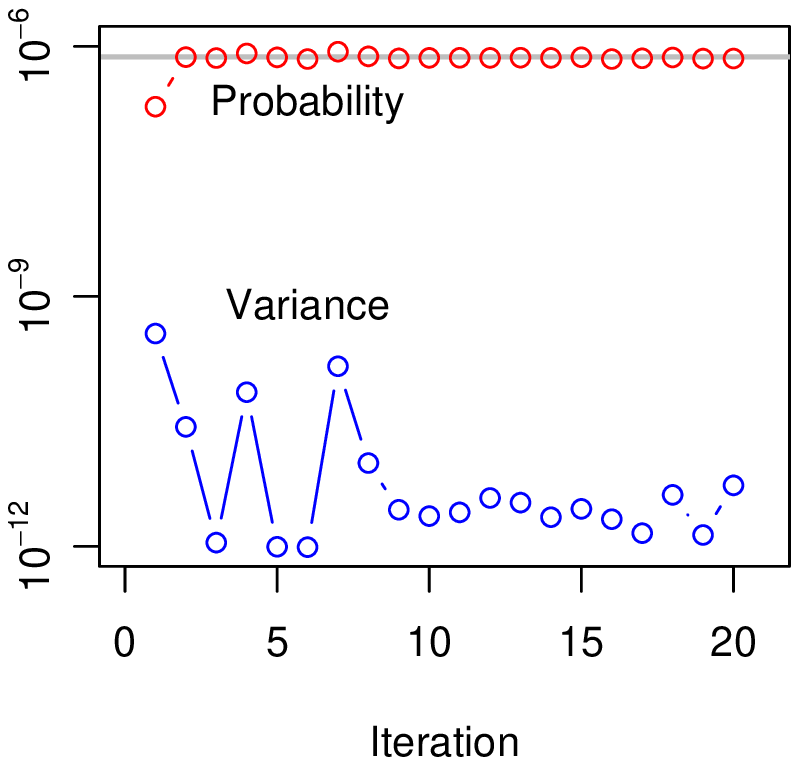}
\caption{Convergence of estimated probability and sample variance for repair model using $N_j=10000$. True probability shown as horizontal line.}
\label{fig:gammaRidder3}
\end{minipage}
\end{figure}

\section{Conclusions and future work}
Statistical model checking addresses the state space explosion associated with exact probabilistic model checking by estimating the parameters of an empirical distribution of executions of a system. By constructing an executable model, rather than an explicit representation of the state space, SMC is able to quantify and verify the performance of systems that are intractable to an exhaustive approach. SMC trades certainty for tractability and often offers the only feasible means to certify real-world systems. Rare properties pose a particular problem to Monte Carlo simulation methods because the properties are difficult to observe and the error in their estimated probabilities is difficult to bound. Importance sampling is a well-established means to reduce the variance of rare events but requires the construction of a suitable importance sampling distribution without resorting to the exploration of the entire state space.

We have devised a natural parametrisation for importance sampling and have provided a simple algorithm, based on cross-entropy minimisation, to optimise the parameters for use in statistical model checking. We have shown that our parametrisation leads to a unique optimum and have demonstrated that with very few parameters our algorithm can make significant improvements in the efficiency of statistical model checking. We have shown that our approach is applicable to standard reliability models and to the kind of huge state space models found in systems biology. We therefore anticipate that our methodology has the potential to be applied to many complex natural and man-made systems.

An ongoing challenge is to find ways to accurately bound the error of results obtained by importance sampling. Specifically, the sample variance of the results may be a very poor indicator of the true variance (i.e. with respect to the unknown true probability). Recent work has addressed this problem using Markov chain coupling applied to a restricted class of models and logic \cite{Barbot2012}, but a simple universal solution remains elusive. A related challenge is to find precise means to judge the quality of the importance sampling distributions we create. Our algorithm finds an optimum based on an automatic parametrisation of a model described in terms of guarded commands. This description is usually derived from a higher level syntactic description that is likely optimised for compactness rather than consideration of importance sampling. As such, there may be alternative, equivalent ways of describing the model that produce better importance sampling distributions. Applying existing work on the robustness of estimators (see, e.g., Chapter 4 in \cite{Rubino2009}), we hope to adapt our algorithm to provide hints about improved parametrisation.

\bibliographystyle{plain}
\bibliography{CAV2012}
\end{document}